\documentclass[12pt,a4paper]{amsart}

\usepackage{amsmath,amsfonts,amssymb}
\usepackage{color}
\usepackage[hmargin=2cm,vmargin=2cm]{geometry}

\usepackage[hypertexnames=false,hyperfootnotes=false,colorlinks=true,linkcolor=blue,%
citecolor=purple,filecolor=magenta,urlcolor=cyan,unicode,linktocpage=true,pagebackref=false]{hyperref}
\usepackage{nameref,zref-xr}

\usepackage[all]{xy}           

\usepackage{cancel}

\usepackage{tikz}

\newcommand{\mbZ}{\mathbb Z}
\newcommand{\mbC}{\mathbb C}

\newcommand{\oM}{\overline{\mathcal M}}

\def\oM{{\overline{\mathcal{M}}}}
\def\CP{{{\mathbb C}{\mathbb P}}}

\def\mbQ{{\mathbb Q}}
\def\d{{\partial}}

\newcommand{\eps}{\varepsilon}

\newcommand{\DR}{\mathrm{DR}}

\DeclareMathOperator{\Aut}{Aut}

\newcommand{\mcW}{\mathcal{W}}

\newcommand{\mcP}{\mathcal{P}}

\def\g{\mathfrak{g}}

\def\d{\partial}
\def\n{\noindent}

\newcommand{\beq}{\begin{equation}}
\newcommand{\eeq}{\end{equation}}
\newcommand{\beqn}{\begin{equation*}}
\newcommand{\eeqn}{\end{equation*}}

\newcommand{\oG}{\overline{G}}
\newcommand{\oH}{\overline{H}}

\newcommand{\oR}{\overline{R}}

\newcommand{\lb}{\left(}
\newcommand{\rb}{\right)}

\newcommand{\oa}{\overline{a}}

\newcommand{\KdV}{\mathrm{KdV}}

\newcommand{\mbR}{\mathbb{R}}
\newcommand{\Pol}{\mathrm{Pol}}

\newtheorem{theorem}{Theorem}[section]
\newtheorem{proposition}[theorem]{Proposition}
\newtheorem{lemma}[theorem]{Lemma}

\newtheorem{conjecture}[theorem]{Conjecture}

\theoremstyle{definition}

\newtheorem{remark}[theorem]{Remark}

\numberwithin{equation}{section}

\begin{document}

\title[Quantum hierarchies of the ILW and BO equations]{Quantum hierarchies of the ILW and BO equations and the DR cycles on the moduli spaces of curves}

\author{Alexandr Buryak}
\address{A. Buryak:\newline 
Faculty of Mathematics, National Research University Higher School of Economics, Usacheva str. 6, Moscow, 119048, Russian Federation;\smallskip\newline 
Skolkovo Institute of Science and Technology, Bolshoy Boulevard 30, bld. 1, Moscow, 121205, Russian Federation}
\email{aburyak@hse.ru}

\author{Abner Hernandez Rodriguez}
\address{A. Hernandez Rodriguez:\newline 
Faculty of Mathematics, National Research University Higher School of Economics, Usacheva str. 6, Moscow, 119048, Russian Federation}
\email{ernandes.r@hse.ru}

\begin{abstract}
In an earlier paper of P. Rossi and the first author, a construction of a wide class of quantum hierarchies was presented, with the input data given by an arbitrary cohomological field theory. A crucial role in the construction is played by the double ramification (DR) cycles on the moduli spaces of stable curves. As a special case, one gets a quantization of the hierarchy of the Intermediate Long Wave (ILW) equation. The ILW equation has two famous limits: the Korteweg--de Vries equation and the Benjamin--Ono (BO) equation, and for the latter one a quantum hierarchy was constructed by Nazarov and Sklyanin. In our paper, we propose a conjecture describing a relation between a limit of the quantum ILW hierarchy and the Nazarov--Sklyanin hierarchy and present various checks to support it.         
\end{abstract}

\date{\today}

\maketitle

\section{Introduction}

The cohomology of the moduli spaces of stable curves $\oM_{g,n}$ is deeply related to the theory of integrable systems. The Witten conjecture~\cite{Wit91}, proved by M. Kontsevich~\cite{Kon92}, provided a bright manifestation of this relation opening a new direction of research. The cohomology of~$\oM_{g,n}$ can be viewed as a tool for a construction of integrable systems, and the crucial role here is played by the notion of a \emph{cohomological field theory (CohFT)} that is a collection of classes in the cohomology of $\oM_{g,n}$ satisfying certain compatibility conditions with respect to gluing and forgetful maps between the moduli spaces. In~\cite{Bur15a}, using remarkable properties of the \emph{double ramification (DR) cycles} $\DR_g(a_1,\ldots,a_n)\in H^{2g}(\oM_{g,n})$, $a_1,\ldots,a_n\in\mbZ$, the author proposed a construction of a Hamiltonian hierarchy of evolutionary PDEs with one spatial variable, called the \emph{DR hierarchy}, associated to an arbitrary CohFT. Many known integrable hierarchies appear as the DR hierarchies corresponding to certain CohFTs or their generalizations. Note, however, that it is still not known how to determine whether a given hierarchy of evolutionary PDEs belongs to the class of DR hierarchies or not, although there are several conjectures in the case when there is only one dependent variable~\cite{BR25,BXY26}. 

\medskip

In~\cite{BR16}, the authors constructed a quantization of the DR hierarchy, called the \emph{quantum DR hierarchy}, corresponding to an arbitrary CohFT. The quantum DR hierarchy is completely determined by a particular Hamiltonian, which we call the \emph{key Hamiltonian}, using a certain recursion procedure~\cite{BR16}. In~\cite{BR16,dCR19,BDGR20}, the key quantum Hamiltonians associated to various CohFTs were computed, providing quantizations of the ILW hierarchy, the extended Toda hierarchy, the Drinfeld--Sokolov hierarchy corresponding to the $D_4$-singularity, the $r$-KdV hierarchy for $r=3$ and $r=4$. In~\cite{RY21,vIR24,vIR25}, the authors studied the eigenvectors and the eigenvalues of the quantum KdV hierarchy. For the quantum DR hierarchies, there is a notion of \emph{quantum tau-function}, which was introduced in~\cite{BDGR20}. Various results and conjectures about it for the quantum KdV hierarchy were proposed in~\cite{Blo22}, for further developments see~\cite{BB24,BLS24}.

\medskip

In the case of trivial CohFT, the corresponding quantum DR hierarchy is a quantization of the KdV hierarchy, and the first two Hamiltonians are
\begin{align*}
&\oG^\KdV_1=\int:\left(\frac{u^3}{6}+\frac{\eps^2}{24}u u_{xx}\right):dx,\\
&\oG^\KdV_2=\int:\left(\frac{u^4}{24}+\eps^2\frac{u^2 u_2}{48}+\eps^4\frac{u u_4}{480}-i\hbar\frac{2uu_2+u^2}{48}\right):dx,
\end{align*}
where $u=\sum_{a\in\mbZ\backslash\{0\}}p_a e^{iax}$, $[p_a,p_{-a}]=i\hbar a$, $\eps$ and $\mu$ are additional formal variables, we use the normal ordering $:\cdot:$ that moves all the $p_a$'s with positive $a$ to the right and all the $p_a$'s with negative $a$ to the left, and the integration sign means that we extract the coefficient of $e^{i0x}$. At the dispersionless level the same quantization already appeared before in~\cite{Pog03,Eli07}. 

\medskip

There is a one-parameter deformation of the trivial CohFT given by the full Chern class of the Hodge bundle over $\oM_{g,n}$. According to~\cite{BR16}, the corresponding quantum DR hierarchy is a quantization of the ILW hierarchy, and the key quantum Hamiltonian is equal to
\begin{align}
\oG_1=&\int:\left(\frac{u^3}{6}+(\eps^2-i\hbar\mu)\sum_{g\ge 1}\eps^{2g-2}\mu^{g-1}\frac{|B_{2g}|}{2(2g)!}uu_{2g}\right):dx=\label{eq:formula for oG1}\\
=&\frac{1}{2}\sum_{a,b\ge 1}(p_{-a-b}p_ap_b+p_{-a}p_{-b}p_{a+b})-(\eps^2-i\hbar\mu)\sum_{g\ge 1}\eps^{2g-2}\mu^{g-1}\frac{B_{2g}}{(2g)!}\sum_{a\ge 1}a^{2g}p_{-a}p_a,\notag
\end{align}
where $B_{2g}$ are the Bernoulli numbers. In this paper we are interested in a limit of $\oG_1$, which is sometimes called the deep water limit. Using that $\sum_{n\ge 0}\frac{B_{2n}}{(2n)!}x^{2n}=\frac{x}{2}\frac{e^x+1}{e^x-1}$, we can express $\oG_1$ in the following way:
\beqn
\oG_1=\frac{1}{2}\sum_{a,b\ge 1}(p_{-a-b}p_ap_b+p_{-a}p_{-b}p_{a+b})+(\eps^2-i\hbar\mu)\sum_{a\ge 1}\frac{1-\frac{a\eps\sqrt\mu}{2}\frac{e^{a\eps\sqrt\mu}+1}{e^{a\eps\sqrt\mu}-1}}{\eps^2\mu}p_{-a}p_{a}.
\eeqn
We see that after considering $\mu$ as a positive real number and making the substitution $\eps=\sqrt\mu$ there is a well-defined limit when $\mu\mapsto+\infty$:
\beqn
\oG_1^{\lim}:=\lim_{\mu\to+\infty}\lb\left.\oG_1\right|_{\eps=\sqrt\mu}\rb=\frac{1}{2}\sum_{a,b\ge 1}(p_{-a-b}p_ap_b+p_{-a}p_{-b}p_{a+b})+\frac{i\hbar-1}{2}\sum_{a\ge 1}a p_{-a}p_a.
\eeqn
The operator $2\oG_1^{\lim}$ coincides with the first nontrivial Hamiltonian of the quantum hierarchy of the Benjamin--Ono (BO) equation from the paper~\cite{NS13}: see the operator~(1.1) there and note that we should identify our variables $p_n$ and $p_{-n}$, $n>0$, with the variables $p_n^*$ and $p_{-n}$, respectively, from~\cite{NS13}, and also identify $i\hbar=\alpha$. It is natural to study whether the limits of the higher Hamiltonians $\oG_d$ of the quantum ILW hierarchy exist, and, if yes, how to relate them to the higher Hamiltonians of the quantum BO hierarchy of Nazarov and Sklyanin.  

\medskip

The main results of the paper are the following:
\begin{enumerate}
\item We construct a triangular linear transformation of the operators $\oG_d':=\left.\oG_d\right|_{\eps=\sqrt\mu}$, after which the resulting operators $\oG_d^{\bullet}$ conjecturally have a limit when $\mu\to+\infty$.  

\smallskip

\item We construct a triangular polynomial transformation of the operators $\lim_{\mu\to+\infty}\oG_d^{\bullet}$, which conjecturally relates them to the operators of the quantum BO hierarchy.

\smallskip

\item We check our conjectures for the operators $\oG_1,\oG_2,\oG_3$.

\smallskip

\item We check our conjectures when $i\hbar=1$.
\end{enumerate}

\medskip

\subsection*{Acknowledgements} This work is an output of the research project HSE-BR-2025-059 implemented as a part of the Basic Research Program at the National Research University Higher School of Economics (HSE University). The authors are grateful to B. Feigin, A. Litvinov, and M. Matushko for useful discussions.

\medskip

\section{The quantum ILW and BO hierarchies and our conjecture}

\subsection{Cohomology classes on the moduli space of curves $\oM_{g,n}$}

We will work with the moduli spaces $\oM_{g,n}$ of stable algebraic curves of genus $g$ with $n$ marked points. Note that $\oM_{g,n}$ is nonempty only when the condition $2g-2+n>0$ is satisfied. We will often omit mentioning this condition explicitly and assume that it is satisfied when a moduli space is considered. We will use the following standard cohomology classes (in singular cohomology with complex coefficients) on $\oM_{g,n}$:
\begin{itemize}
\item The psi-class $\psi_i\in H^2(\oM_{g,n})$, $1\le i\le n$, is the first Chern class of the line bundle over~$\oM_{g,n}$ formed by the cotangent lines at the $i$-th marked point of stable curves.

\smallskip

\item The Hodge class $\lambda_j:= c_j(\mathbb E)\in H^{2j}(\oM_{g,n})$, $j\ge 0$, where~$\mathbb E$ is the rank~$g$ Hodge vector bundle over~$\oM_{g,n}$ whose fibers are the spaces of holomorphic one-forms on stable curves. We define $\Lambda(\mu):=\sum_{i=0}^g\lambda_i\mu^i$. We will use the property $\Lambda(\mu)\Lambda(-\mu)=1$.

\smallskip

\item The \emph{double ramification (DR) cycle} $\DR_g(a_1,\ldots,a_n)\in H^{2g}(\oM_{g,n})$, $a_1,\ldots,a_n\in\mbZ$, $\sum a_i=0$, is defined as follows. There is a moduli space of projectivized stable maps to~$\CP^1$ relative to~$0$ and~$\infty$, with ramification profile over $0$ given by the negative numbers among the $a_i$-s, ramification profile over $\infty$ given by the positive numbers among the $a_i$-s, and the zeros among the $a_i$-s correspond to additional marked points (see, e.g.,~\cite{BSSZ15} for more details). This moduli space is endowed with a virtual fundamental class, which lies in the homology group of degree $2(2g-3+n)$. The DR cycle~$\DR_g(a_1,\ldots,a_n)$ is the Poincar\'e dual to the pushforward, through the forgetful map to~$\oM_{g,n}$, of this virtual fundamental class. An important property of the DR cycle is that $\DR_g\left(-\sum b_i,b_1,\ldots,b_n\right)$ is an even polynomial in $b_1,\ldots,b_n$ of degree~$2g$ (nonhomogeneous in general) with the coefficients from the vector space $H^{2g}(\oM_{g,n+1})$~\cite{Spe26}. Note also that $\DR_0(a_1,\ldots,a_n)=1\in H^0(\oM_{0,n})$.
\end{itemize}

\medskip

\subsection{Quantum ILW hierarchy}

Consider formal variables $q_1,q_2,\ldots$. The set of all partitions will be denoted by $\mcP$. For a partition $\lambda=(\lambda_1,\ldots,\lambda_l)$, let 
$$
q_\lambda:=\prod_{i=1}^l q_{\lambda_i},\quad \frac{\d}{\d q_\lambda}:=\frac{\d^l}{\d q_{\lambda_1}\cdots\d q_{\lambda_l}},\quad m_k(\lambda):=|\{1\le i\le l|\lambda_i=k\}|,\quad \Aut(\lambda):=\prod_{i\ge 1}m_i(\lambda)!.
$$
We denote by $\mcW$ the algebra of operators of the form
\beq\label{eq:operators}
\sum_{\lambda\in\mcP}q_\lambda P_\lambda\lb\frac{\d}{\d q_1},\frac{\d}{\d q_2},\ldots\rb,
\eeq
where $P_\lambda$ are polynomials with complex coefficients. These operators act in the algebra $\mbC[[q_1,q_2,\ldots]]$. For an integer $p\ge 0$, denote by $\mcW^{(p)}$ the vector subspace of $\mcW$ formed by the operators of the form
\beqn
\sum_{\substack{\lambda,\mu\in\mcP\\l(\lambda)+l(\mu)=p}}c_{\lambda,\mu} q_\lambda \frac{\d}{\d q_\mu},\quad c_{\lambda,\mu}\in\mbC,
\eeqn
and let $\mcW^{(\le p)}:=\oplus_{i=0}^p\mcW^{(i)}$. Denote by $\pi_p\colon\mcW\to\mcW^{(p)}$ the natural projection. We will also consider operators of the form~\eqref{eq:operators} where the coefficients of the polynomials $P_\lambda$ belong to the algebra of formal power series in some variables, say $\eps_1,\eps_2,\ldots,\eps_k$. The algebra of such operators will be denoted by $\mcW_{\eps_1,\ldots,\eps_k}$.

\medskip

The operators of the quantum ILW hierarchy belong to the algebra $\mcW_{\hbar,\eps,\mu}$ and are defined as follows:
\beqn
\oG_d:=\hspace{-0.3cm}\sum_{\substack{n\ge 2\\a_1,\ldots,a_n\in\mbZ\\\sum a_i=0}}\sum_{g\ge 0}\frac{(i\hbar)^g}{n!}\left(\int_{\oM_{g,n+1}}\hspace{-0.3cm}\DR_g\left(0,a_1,\ldots,a_n\right)\Lambda\left(\frac{-\eps^2}{i \hbar}\right)\Lambda(\mu) \psi_1^d\right):p_{a_1}\cdots p_{a_n}:,\quad d\ge 0.
\eeqn
where
$$
p_a:=
\begin{cases}
q_{-a},&\text{if $a<0$},\\
ia\hbar\frac{\d}{\d q_a},&\text{if $a>0$},\\
0,&\text{if $a=0$}.
\end{cases}
$$
For example, $\oG_0=\frac{1}{2}\sum_{a\in\mbZ\backslash\{0\}}:p_a p_{-a}:=\sum_{a\ge 1}q_a p_a$, and the operator $\oG_1$ is given by formula~\eqref{eq:formula for oG1}, which follows from the equality (see, e.g.,~\cite[Section 4.2]{BR16})
\beqn
\int_{\oM_{g,3}}\DR_g(0,a,-a)\psi_1\lambda_g\lambda_{g-1}=a^{2g}\frac{|B_{2g}|}{(2g)!},\quad g\ge 1.
\eeqn
The operators $\oG_d$ pairwise commute:
$$
[\oG_{d_1},\oG_{d_2}]=0,\quad d_1,d_2\ge 0.
$$

\medskip

Note that the formula for $\oG_d$ involves the integrals
$$
\int_{\oM_{g,n+1}}\DR_g\left(0,a_1,\ldots,a_n\right)\lambda_{g-j}\lambda_{g-k}\psi_1^d,
$$
which by degree counting is zero unless $j+k=d+2-n$. So the operator $\oG_d$ has the form
$$
\oG_d=\hspace{-0.3cm}\sum_{\substack{2\le n\le d+2\\a_1,\ldots,a_n\in\mbZ\\\sum a_i=0}}\sum_{\substack{j+k=d+2-n\\g\ge\max(j,k)}}\hspace{-0.3cm}\frac{(-\eps^2)^{g-j}(i\hbar)^j\mu^{g-k}}{n!}\left(\int_{\oM_{g,n+1}}\hspace{-0.3cm}\DR_g(0,a_1,\ldots,a_n)\lambda_{g-j}\lambda_{g-k}\psi_1^d\right):p_{a_1}\cdots p_{a_n}:,
$$
and, in particular, it belongs to $\mcW^{(\le d+2)}_{\hbar,\eps,\mu}$. Because of the polynomiality property of the DR cycles, the operator $\oG_d$ has the form
\beq\label{eq:oGd in terms of differential polynomials}
\oG_d=\int:\frac{u^{d+2}}{(d+2)!}:dx+\sum_{\substack{1\le j+k\le d\\g\ge\max(j,k)}}(-\eps^2)^{g-j}(i\hbar)^j\mu^{g-k}\int:\frac{Q_{d,g,j,k}(u,u_x,u_{xx},\ldots)}{(d+2-j-k)!}: dx,
\eeq
where $Q_{d,g,j,k}$ is a polynomial in $u,u_x,u_{xx},\ldots$ that is a linear combination of monomials of the form $u_{d_1}\cdots u_{d_n}$, where $u_d:=\d_x^d u$, $n=d+2-j-k$, and $\sum d_i\le 2g$.

\medskip

\subsection{Quantum BO hierarchy of Nazarov and Sklyanin}

In~\cite{NS13}, the authors introduced the following operators:
$$
I^{(k)}:=\sum_{i_1,\ldots,i_k\ge 1}q_{i_1}L_{i_1i_2}\cdots L_{i_{k-1}i_k}p_{i_k},\quad k\ge 1,
$$
where 
$$
L_{ij}:=
\begin{cases}
p_{i-j},&\text{if $i\ne j$},\\
(\alpha-1)i,&\text{if $i=j$},
\end{cases}\qquad \alpha:=i\hbar.
$$
Note that in~\cite{NS13} the authors denote $q^*_n:=\alpha n\frac{\d}{\d q_n}$, $n\ge 1$. For example,
\begin{align}
&I^{(1)}=\sum_{n\ge 1}q_n p_n,\notag\\
&I^{(2)}=\sum_{a,b\ge 1}(q_a q_b p_{a+b}+q_{a+b}p_{a}p_{b})+(\alpha-1)\sum_{a\ge 1}a q_a p_{a},\notag\\
&I^{(3)}=\sum_{i,j,k\ge 1}q_i p_{i-j}p_{j-k}p_k+(\alpha-1)\sum_{i,j\ge 1}(i+j)q_{i}p_{i-j}p_j+(\alpha-1)^2\sum_{i\ge 1}i^2 q_i p_i,\label{eq:formula for I3}\\
&I^{(4)}=\sum_{i,j,k,l\ge 1}q_i p_{i-j}p_{j-k}p_{k-l}p_l+(\alpha-1)\sum_{i,j,k\ge 1}(i+j+k)q_i p_{i-j}p_{j-k}p_k+\notag\\
&\hspace{1.1cm}+(\alpha-1)^2\sum_{i,j\ge 1}(i^2+ij+j^2)q_{i}p_{i-j}p_j+(\alpha-1)^3\sum_{i\ge 1}i^3 q_i p_i,\notag
\end{align}
where we recall that $p_0=0$. By~\cite[Theorem 2]{NS13}, the Hamiltonians $I^{(k)}$ pairwise commute: 
$$
[I^{(k_1)},I^{(k_2)}]=0,\quad k_1,k_2\ge 1.
$$

\medskip

\subsection{Conjecture}

Let us denote
$$
\oG_d':=\oG_d\big|_{\eps\mapsto\sqrt\mu}.
$$
We see that $\oG'_0=\oG_0=I^{(1)}=\sum_{a\ge 1}q_a p_a$ and 
\beqn
\oG'_1=\frac{1}{2}\sum_{a,b\ge 1}(q_{a+b}p_ap_b+q_{a}q_{b}p_{a+b})+(1-i\hbar)\sum_{a\ge 1}\lb\mu^{-1}-\frac{a}{2}\frac{e^{a\mu}+1}{e^{a\mu}-1}\rb q_a p_a.
\eeqn
In general, the operator $\oG_d'$ has the form
$$
\oG_d'=\sum_{j=0}^d(i\hbar)^j\sum_{n=2}^{d+2-j}\sum_{\substack{\lambda,\rho\in\mcP\\l(\lambda)+l(\rho)=n\\|\lambda|=|\rho|}}h_{d,j,\lambda,\rho}(\mu)q_\lambda p_\rho,
$$
where $h_{d,j,\lambda,\rho}(\mu)\in\mbQ[[\mu]]$. We see that for $d=1$ the formal power series $h_{1,j,\lambda,\rho}(\mu)$ is the power series expansion at $0$ of a function from $\mbQ(\mu,e^{\mu})$, and moreover this function has a limit when $\mu\in\mbR$ and $\mu\to+\infty$. However, for $d\ge 3$ it can happen that the radius of convergence of~$h_{d,j,\lambda,\rho}(\mu)$ is zero, and in this case we don't know a reasonable way to take the limit when $\mu\to+\infty$. In Section~\ref{subsection:G3 Hamiltonian} we will compute $\oG_3'$ and show it, see Remark~\ref{remark:divergent term}. However, here we want to present a short computation that we believe is related to the problem with a limit and also shows how to overcome it. 

\medskip

By~\eqref{eq:oGd in terms of differential polynomials}, we have
$$
\oG_d'=\int:\frac{u^{d+2}}{(d+2)!}:dx+\sum_{\substack{1\le j+k\le d\\g\ge\max(j,k)}}(-1)^{g-j}(i\hbar)^j\mu^{2g-j-k}\int:\frac{Q_{d,g,j,k}}{(d+2-j-k)!}:dx.
$$
Let $Q_{d,g,j,k}^\circ:=\left.Q_{d,g,j,k}\right|_{u_{\ge 1}=0}$ and 
$$
\oG_d^\circ:=\int:\frac{u^{d+2}}{(d+2)!}:dx+\sum_{\substack{1\le j+k\le d\\g\ge\max(j,k)}}(-1)^{g-j}(i\hbar)^j\mu^{2g-j-k}\int:\frac{Q^\circ_{d,g,j,k}}{(d+2-j-k)!}:dx.
$$
Using that $\DR_g(0,\ldots,0)=(-1)^g\lambda_g$, we compute
\begin{align*}
\oG_d^\circ=&\int:\frac{u^{d+2}}{(d+2)!}:dx+\sum_{\substack{j,k\ge 1\\j+k\le d\\g\ge\max(j,k)}}(-i\hbar)^j\mu^{2g-j-k}\int_{\oM_{g,1}}\lambda_g\lambda_{g-j}\lambda_{g-k}\psi_1^{j+k-2}\int:\frac{u^{d+2-j-k}}{(d+2-j-k)!}:dx=\\
=&\int:\frac{u^{d+2}}{(d+2)!}:dx+\sum_{\substack{j,k\ge 1\\j+k\le d}}(-i\hbar)^jf_{j,k}(\mu)\int:\frac{u^{d+2-j-k}}{(d+2-j-k)!}:dx,
\end{align*}
where
$$
f_{j,k}(\mu):=\sum_{g\ge\max(j,k)}\mu^{2g-j-k}\int_{\oM_{g,1}}\lambda_g\lambda_{g-j}\lambda_{g-k}\psi_1^{j+k-2},\quad j,k\ge 1.
$$
For $j=k=1$, we have $f_{1,1}(\mu)=\sum_{g\ge 1}\mu^{2g-2}\int_{\oM_{g,1}}\lambda_g\lambda_{g-1}^2=\frac{1}{24}$. However, we believe that if $\max(j,k)\ge 2$, then the radius of convergence of the formal power series $f_{j,k}(\mu)$ is zero. For example, using that (see~\cite[eq.~(9)]{FP00})
$$
\int_{\oM_g}\lambda_g\lambda_{g-1}\lambda_{g-2}=\frac{1}{2(2g-2)!}\frac{|B_{2g-2}|}{2g-2}\frac{|B_{2g}|}{2g},\quad g\ge 2,
$$
we compute
$$
f_{1,2}(\mu)=\sum_{g\ge 2}\frac{|B_{2g-2}||B_{2g}|}{4g(2g-2)!}\mu^{2g-3}.
$$
Since $|B_{2n}|=\frac{2(2n)!}{(2\pi)^{2n}}\zeta(2n)\geq \frac{(2n)!}{(2\pi)^{2n}}$, $n\geq 1$, we obtain
$$
\frac{|B_{2g-2}||B_{2g}|}{4g(2g-2)!}\ge \frac{2(2g-1)!}{(2\pi)^{4g-2}},
$$
and therefore the radius of convergence of $f_{1,2}(\mu)$ is zero. This motivates us to make a triangular linear combination of the Hamiltonians $\oG_d'$ in order to kill all the terms $\int:\frac{u^{d+2-j-k}}{(d+2-j-k)!}:dx$.

\medskip

Now we are ready to formulate a precise conjecture. Define formal power series $A,B\in 1+z^2\mbQ[[i\hbar,\mu,z]]$ by
$$
A:=1+\sum_{j,k\ge 1}(-i\hbar)^jf_{j,k}(\mu)z^{j+k},\qquad B:=\frac{1}{A}=1+\sum_{d\ge 2}b_d(\hbar,\mu)z^d,
$$
and define \emph{normalized Hamiltonians} $\oG_d^\bullet$ by
$$
\oG_d^\bullet:=\oG_d'+\sum_{i=2}^d b_d(\hbar,\mu)\oG'_{d-i},\quad d\ge 0.
$$
For example, we have
\begin{align*}
&\oG_0^{\bullet}=\oG_0',\\
&\oG_1^{\bullet}=\oG_1',\\
&\oG_2^{\bullet}=\oG_2'+\frac{i\hbar}{24}\oG_0',\\
&\oG_3^{\bullet}=\oG_3'+\frac{i\hbar}{24}\oG_1'+i\hbar(1-i\hbar)f_{1,2}(\mu)\oG_0'.
\end{align*}
Note that $b_2(\hbar,\mu)=\frac{i\hbar}{24}$ and $\oG_d^{\bullet}=\oG_d'+\frac{i\hbar}{24}\oG_{d-2}'+\ldots$. In general, the normalized Hamiltonian~$\oG_d^\bullet$ has the form
$$
\oG_d^\bullet=\sum_{j=0}^d(i\hbar)^j\sum_{n=2}^{d+2-j}\sum_{\substack{\lambda,\rho\in\mcP\\l(\lambda)+l(\rho)=n\\|\lambda|=|\rho|}}h^\bullet_{d,j,\lambda,\rho}(\mu)q_\lambda p_\rho,\quad h^\bullet_{d,j,\lambda,\rho}(\mu)\in\mbQ[[\mu]].
$$

\medskip

\begin{conjecture}{\ }
\begin{enumerate}
\item Each formal power series $h^\bullet_{d,j,\lambda,\rho}(\mu)$ is the power series expansion at $0$ of a function from~$\mbQ(\mu,e^\mu)$. Moreover, this function has a limit, when $\mu\in\mbR$ and $\mu\to+\infty$. 

\smallskip

\item Denote $\oG_d^{\lim}:=\lim_{\mu\in\mbR,\,\mu\to+\infty}\oG_d^\bullet$. Then we have the following relation between generating series:
$$
\sum_{k\ge 1}I^{(k)}u^{-k}=u\bigg( 1-\exp\bigg(-\sum_{d\ge 2}(d-1)!\oG^{\lim}_{d-2}u^{-d}\bigg)\bigg).
$$
\end{enumerate}
\end{conjecture}

\medskip

For example, the conjecture claims that
\beq\label{eq:I and G for d2 and d3}
I^{(3)}=3!\oG^{\lim}_2-\frac{1}{2}(\oG^{\lim}_0)^2,\qquad I^{(4)}=4!\oG_3^{\lim}- 2 \oG_0^{\lim} \oG_1^{\lim}.
\eeq
In our paper, we check the first part of the conjecture for $d\le 3$ and prove these two equations. We also prove the conjecture when $i\hbar=1$.

\medskip

\section{Proof of the conjecture in the case $\alpha=i\hbar=1$}

Throughout this section, we assume that $i\hbar=\alpha=1$. Then, since $\Lambda(\mu)\Lambda(-\mu)=1$, we have 
\beqn
\oG'_d:=\hspace{-0.3cm}\sum_{\substack{n\ge 2\\a_1,\ldots,a_n\in\mbZ\\\sum a_i=0}}\sum_{g\ge 0}\frac{1}{n!}\left(\int_{\oM_{g,n+1}}\DR_g\left(0,a_1,\ldots,a_n\right) \psi_1^d\right):p_{a_1}\cdots p_{a_n}:,\quad d\ge 0.
\eeqn
We see that the operators $\oG_d'$ don't depend on $\mu$, so we don't have to consider the procedure of taking a limit.
Let us first relate the operators $\oG'_k$ and $I^{(l)}$. For any $d\ge 0$, these operators preserve the subspace $\Pol_d\subset\mbC[[q_1,q_2,\ldots]]$ of polynomials of degree $d$, where $\deg q_i:=i$. We will use the fact that the operators $\oG'_k$ and $I^{(l)}$ become diagonal in the basis in $\Pol_d$ given by the Schur symmetric functions~$s_\lambda$, $\lambda\in\mcP_d$:
$$
s_\lambda:=\det(h_{\lambda_j-j+i})_{1\le i,j\le l(\lambda)},\qquad \sum_{j\ge 0} h_j z^j:=\exp\bigg(\sum_{i\ge 1}\frac{q_i}{i}z^i\bigg),\qquad h_j:=0\text{ for $j<0$}.
$$

\medskip

On one hand, we have (see e.g.~\cite[eq.~(1.14)]{vIR24})
$$
\oG'_{k}s_\lambda=\underbrace{\frac{1}{(k+1)!}\sum_{i\ge 1}\left[\lb\lambda_i-i+\frac{1}{2}\rb^{k+1}-\lb-i+\frac{1}{2}\rb^{k+1}\right]}_{E_k(\lambda):=}s_\lambda.
$$
On the other hand, we have (see \cite[Section 6]{NS13})
$$
\lb\sum_{k\ge 1}I^{(k)}u^{-k}\rb s_\lambda=\underbrace{\left[u-(u+l(\lambda))\prod_{i=1}^{l(\lambda)}\frac{u+i-1-\lambda_i}{u+i-\lambda_i}\right]}_{\sum_{k\ge 1}F_k(\lambda)u^{-k}:=}s_\lambda.
$$
We compute (we denote $l=l(\lambda)$)
\begin{multline}\label{eq:tmp}
\log\left[(1+u^{-1}l)\prod_{i=1}^l\frac{1-u^{-1}(\lambda_i-i+1)}{1-u^{-1}(\lambda_i-i)}\right]=\\
=\sum_{n\ge 1}\frac{u^{-n}}{n}\lb(-1)^{n-1}l^n+\sum_{i=1}^l\lb(\lambda_i-i)^n-(\lambda_i-i+1)^n\rb\rb
\end{multline}
and proceed as follows:
\begin{align*}
&\sum_{i=1}^l\lb\lambda_i-i+\frac{1}{2}\pm\frac{1}{2}\rb^n=\sum_{i=1}^l\sum_{m=0}^n\lb\lambda_i-i+\frac{1}{2}\rb^{n-m}\lb\pm\frac{1}{2}\rb^m{n\choose m}=\\
&\hspace{1.5cm}=\sum_{m=0}^n(n-m)!E_{n-m-1}(\lambda)\lb\pm\frac{1}{2}\rb^m{n\choose m}+\sum_{i=1}^l\sum_{m=0}^n\lb-i+\frac{1}{2}\rb^{n-m}\lb\pm\frac{1}{2}\rb^m{n\choose m}=\\
&\hspace{1.5cm}=\sum_{m=0}^n(n-m)!E_{n-m-1}(\lambda)\lb\pm\frac{1}{2}\rb^m{n\choose m}+\sum_{i=1}^l\lb-i+\frac{1}{2}\pm\frac{1}{2}\rb^n,
\end{align*}
where $E_{-1}(\lambda):=0$. Finally, for~\eqref{eq:tmp} we obtain
\beqn
-\sum_{\substack{n\ge 1,\,m\ge 0\\n\ge 2m+1}}\frac{u^{-n}}{n}E_{n-2m-2}(\lambda)\frac{(n-2m-1)!}{2^{2m}}{n\choose 2m+1}=-\sum_{\substack{n\ge 1,\,m\ge 0\\n\ge 2m+1}}u^{-n}\frac{(n-1)!}{(2m+1)!2^{2m}}E_{n-2m-2}(\lambda)
\eeqn
and conclude that
\beqn 
\sum_{k\ge 1}I^{(k)}u^{-k} = u\bigg(1-\exp\bigg(-\sum_{\substack{n\ge 1,\,m\ge 0\\n\ge 2m+1}}u^{-n}\frac{(n-1)!}{(2m+1)!2^{2m}}G'_{n-2m-2}\bigg)\bigg). 
\eeqn

\medskip

It remains to check that
$$
\oG^{\bullet}_{d-2}=\sum_{m=0}^{\lfloor\frac{d}{2}\rfloor-1}\frac{1}{(2m+1)!2^{2m}}G'_{d-2m-2}.
$$
Indeed, we have
$$
A=1+\sum_{j,k\ge 1}(-1)^j f_{j,k}(\mu)z^{j+k}=1+\sum_{d\ge 2,\,g\ge 1}z^d\mu^{2g-d}\sum_{\substack{j,k\ge 1\\j+k=d}}(-1)^j\int_{\oM_{g,1}}\lambda_g\lambda_{g-j}\lambda_{g-k}\psi_1^{d-2},
$$
where $\lambda_i:=0$ for $i<0$. Since 
$$
\sum_{\substack{j,k\ge 1\\j+k=d}}(-1)^j\lambda_g\lambda_{g-j}\lambda_{g-k}=
\begin{cases}
(-1)^g\lambda_g,&\text{if $d=2g$},\\
0,&\text{otherwise},
\end{cases}  
$$
we obtain
$$
A=1+\sum_{g\ge 1}z^{2g}(-1)^g\int_{\oM_{g,1}}\lambda_g\psi_1^{2g-2}=1+\sum_{g\ge 1}\frac{1-2^{2g-1}}{2^{2g-1}}\frac{B_{2g}}{(2g)!}z^{2g}=\frac{z}{e^{z/2}-e^{-z/2}},
$$
where in the middle equation we used~\cite[Theorem~1]{FP03}. Therefore, $B=\frac{e^{z/2}-e^{-z/2}}{z}=\sum_{m\ge 0}\frac{z^{2m}}{2^{2m}(2m+1)!}$, which completes the proof of the conjecture in the case $i\hbar=1$.

\medskip


\section{Leading terms of the operators $\oG_d'$}

The goal of this section is to give an explicit formula for $\oG_d^\bullet$ up to terms from $\mcW^{(\le d)}_{\hbar,\mu}$ for an arbitrary~$d$, and up to terms from $\mcW^{(\le d-1)}_{\hbar,\mu}$ for $d=2$ and $d=3$. First, we need the following technical lemma.

\medskip

\begin{lemma}\label{lemma:commutator lemma}
Consider operators $\oG$ and $\oH$ of the form
\begin{align*}
&\oG=\sum_{a_1+a_2+a_3=0}\frac{:p_{a_1}p_{a_2}p_{a_3}:}{6}+\sum_{j=0}^M(i\hbar)^j\sum_{a\ge 1}\sum_{k\ge 0}\mu^k S_{j,k}^{(1)}(a)q_a p_a,\\
&\oH=\sum_{a_1+\ldots+a_d=0}\frac{:p_{a_1}\cdots p_{a_{d}}:}{d!}+\sum_{l=L}^{d-1}\sum_{a_1+\ldots+a_l=0}\sum_{j=0}^N(i\hbar)^j\sum_{k\ge 0}\mu^k S_{l,j,k}^{(2)}(a_1,\ldots,a_l):p_{a_1}\cdots p_{a_l}:,
\end{align*}
where $d\ge 3$, $M,N\ge 0$, $2\le L\le d-1$, $S_{j,k}^{(1)}\colon\mbZ_{\ge 1}\to\mbC$ is an arbitrary function, $S_{l,j,k}^{(2)}(a_1,\ldots,a_l)\in\mbC[a_1,\ldots,a_l]$, and the polynomials $S_{l,j,k}^{(2)}(a_1,\ldots,a_l)$ don't contain constant terms. Suppose that $[\oG,\oH]\in\mcW^{(\le L)}_{\hbar,\mu}$. Then the operator $\oH$ is uniquely determined by~$\oG$.
\end{lemma}
\begin{proof}
By linearity, we have to check that if an operator $\oH'$ of the form
$$
\oH'=\sum_{l=L}^{d-1}\sum_{a_1+\ldots+a_l=0}\sum_{j=0}^N(i\hbar)^j\sum_{k\ge 0}\mu^k S_{l,j,k}^{(3)}(a_1,\ldots,a_l):p_{a_1}\cdots p_{a_l}:,
$$ 
where $S_{l,j,k}^{(3)}(a_1,\ldots,a_l)$ are polynomials without constant terms, satisfies the relation $[\oG,\oH']\in\mcW^{(\le L)}_{\hbar,\mu}$, then $\oH'$ must be zero. Suppose the converse, that $\oH'$ is nonzero. Let $l_0$ be the maximal~$l$ such that the operator 
$$
\sum_{a_1+\ldots+a_l=0}\sum_{j=0}^N(i\hbar)^j\sum_{k\ge 0}\mu^k S_{l,j,k}^{(3)}(a_1,\ldots,a_l):p_{a_1}\cdots p_{a_l}:
$$
is nonzero. Then let $j_0$ be the minimal $j$ such that the operator
$$
\sum_{a_1+\ldots+a_{l_0}=0}\sum_{k\ge 0}\mu^k S_{l_0,j,k}^{(3)}(a_1,\ldots,a_{l_0}):p_{a_1}\cdots p_{a_{l_0}}:
$$
is nonzero. Finally, let $k_0$ be the minimal $k$ such that the operator
$$
\sum_{a_1+\ldots+a_{l_0}=0}S_{l_0,j_0,k}^{(3)}(a_1,\ldots,a_{l_0}):p_{a_1}\cdots p_{a_{l_0}}:
$$
is nonzero. Denote 
\beqn
\oR:=\sum_{a_1+\ldots+a_{l_0}=0}S_{l_0,j_0,k_0}^{(3)}(a_1,\ldots,a_{l_0}):p_{a_1}\cdots p_{a_{l_0}}:,\qquad \oG':=\sum_{a_1+a_2+a_3=0}\frac{:p_{a_1}p_{a_2}p_{a_3}:}{6}.
\eeqn

\medskip

Then we have $\{\oG',\oR\}=0$, where we consider~$\oG'$ and~$\oR$ as formal power series in the variables $p_a$ and $\{\cdot,\cdot\}$ is the standard Poisson bracket given by $\{p_a,p_b\}=ia\delta_{a+b,0}$, i.e., $\{\oG',\oR\}=\sum_{a\in\mbZ_{\ne 0}}ia\frac{\d\oG'}{\d p_a}\frac{\d\oR}{\d p_{-a}}$. We have $\oG'=\int\frac{u^3}{6}dx$ and $\oR=\int q(u,u_x,\ldots)dx$, where $q$ is polynomial in $u,u_x,\ldots$ satisfying $q|_{u_{\ge 1}=0}$. Because of this property of $q$, the relation $\left\{\int\frac{u^3}{6}dx,\int q(u,u_x,\ldots)dx\right\}=0$ implies that $\int q(u,u_x,\ldots)dx=0$ (see, e.g., \cite[Lemma~2.4]{Bur15b} and also~\cite[Section~2.3]{Bur15a} for more details about the Poisson bracket~$\{\cdot,\cdot\}$). This contradiction proves the lemma. 
\end{proof}

\medskip

Let us introduce the following notations:
\begin{align*}
&\phi(x):=\sum_{n\ge 1}\frac{B_{2n}}{(2n)!}x^{2n}=\frac{x}{2}\frac{e^x+1}{e^x-1}-1,\\
&\theta(x):=\sum_{n\ge 2}(n+1)\frac{B_{2n}}{(2n)!}x^{2n}=\frac{\phi(x)-\phi(x)^2}{2}-\frac{x^2}{24}.
\end{align*}
We will denote $[n]:=\{1,\ldots,n\}$.

\medskip

\begin{proposition}\label{proposition:leading terms of G_d}{\ }
\begin{enumerate}
\item For an arbitrary $d\ge 1$, we have
\beqn
\oG_d^\bullet=\sum_{\substack{\oa\in\mbZ^{d+2}\\\sum a_i=0}}\frac{:p_{a_1}\cdots p_{a_{d+2}}:}{(d+2)!}+\frac{i\hbar-1}{\mu}\sum_{\substack{\oa\in\mbZ^{d+1}\\\sum a_i=0}}f_d(\oa,\mu)\frac{:p_{a_1}\cdots p_{a_{d+1}}:}{(d+1)!}+R_d,
\eeqn
where $R_d\in\mcW^{(\le d)}_{\hbar,\mu}$ and
\beqn
f_d(\oa,\mu)=\frac{1}{2}\sum_{I\subset[d+1]}\frac{(|I|-1)!(d-|I|)!}{d!}\phi(a_I\mu),\qquad a_I:=\sum_{i\in I}a_i.
\eeqn

\smallskip

\item For $d=2,3$, we have
\beqn
\pi_d(\oG_d^\bullet)=\sum_{\substack{\oa\in\mbZ^{d}\\\sum a_i=0}}\left(i\hbar\frac{\sum a_i^2}{24}-\frac{(i\hbar-1)^2}{\mu^2}h_d(\oa,\mu)\right)\frac{:p_{a_1}\cdots p_{a_{d}}:}{d!},
\eeqn
where
\begin{align*}
&h_2(a,-a,\mu)=\theta(a\mu)=\frac{\phi(a\mu)-\phi(a\mu)^2}{2}-\frac{(a\mu)^2}{24},\\
&h_3(a_1,a_2,a_3,\mu)=\sum_{i=1}^3\lb\frac{\phi(a_i\mu)}{4}-\frac{\phi(a_i\mu)^2}{6}-\frac{(a_i\mu)^2}{48}\rb - \frac{1}{6}\sum_{1\le i<j\le 3}\phi(a_i\mu)\phi(a_j\mu).
\end{align*}
\end{enumerate}
\end{proposition}
\begin{proof}
\emph{1}. By Lemma~\ref{lemma:commutator lemma}, it is sufficient to check that
$$
\left[\oG_1',\sum_{a_1+\ldots+a_{d+2}=0}\frac{:p_{a_1}\cdots p_{a_{d+2}}:}{(d+2)!}+\frac{i\hbar-1}{\mu}\sum_{a_1+\ldots+a_{d+1}=0}f_d(\oa,\mu)\frac{:p_{a_1}\cdots p_{a_{d+1}}:}{(d+1)!}\right]\in\mcW_{\hbar,\mu}^{(\le d+1)}.
$$
Equivalently, we have to check that
\begin{multline*}
\left\{\sum_{b_1+b_2+b_3=0}\frac{:p_{b_1}p_{b_2}p_{b_3}:}{6},\sum_{a_1+\ldots+a_{d+1}=0}f_d(\oa,\mu)\frac{:p_{a_1}\cdots p_{a_{d+1}}:}{(d+1)!}\right\}+\\
+\left\{\sum_{a\ge 1}\phi(a\mu)p_{-a}p_{a},\sum_{a_1+\ldots+a_{d+2}=0}\frac{:p_{a_1}\cdots p_{a_{d+2}}:}{(d+2)!}\right\}=0.
\end{multline*}
The left-hand side is equal to
\begin{align*}
&-i\hbar\sum_{a_1+\ldots+a_{d+2}=0}\frac{:p_{a_1}\cdots p_{a_{d+2}}:}{(d+2)!}\sum_{1\le i<j\le d+2}(a_i+a_j)f_d\lb a_1,\ldots,\widehat{a_i},\ldots,\widehat{a_j},\ldots,a_{d+2},a_i+a_j,\mu\rb\\
&-i\hbar\sum_{a_1+\ldots+a_{d+2}=0}\frac{:p_{a_1}\cdots p_{a_{d+2}}:}{(d+2)!}\sum_{i=1}^{d+2}a_i\phi(a_i\mu).
\end{align*}
So we have to check that
$$
\sum_{1\le i<j\le d+2}(a_i+a_j)f_d\lb a_1,\ldots,\widehat{a_i},\ldots,\widehat{a_j},\ldots,a_{d+2},a_i+a_j,\mu\rb+\sum_{i=1}^{d+2}a_i\phi(a_i\mu)=0.
$$
Indeed, the left-hand side is equal to
\begin{align*}
&\frac{1}{2}\sum_{\substack{I\subset[d+2]\\i<j\\i,j\notin I}}(a_i+a_j)\left(\frac{(|I|-1)!(d-|I|)!}{d!}\phi(a_I\mu)+\frac{|I|!(d-1-|I|)!}{d!}\phi(a_{I\cup\{i,j\}}\mu)\right)+\sum_{i=1}^{d+2}a_i\phi(a_i\mu)=\\
=&\sum_{\substack{I\subset[d+2]\\2\le|I|\le d}}\frac{(d+1-|I|)a_{[d+2]\backslash I}}{2}\frac{(|I|-1)!(d-|I|)!}{d!}\phi(a_I\mu)+\sum_{i=1}^{d+2}\frac{d\cdot a_{[d+2]\backslash\{i\}}}{2}\frac{(d-1)!}{d!}\phi(a_i\mu)\\
&+\sum_{\substack{I\subset[d+2]\\2\le|I|\le d}}\frac{(|I|-1)a_I}{2}\frac{(|I|-2)!(d+1-|I|)!}{d!}\phi(a_I\mu)+\sum_{i=1}^{d+2}\frac{d\cdot a_{[d+2]\backslash\{i\}}}{2}\frac{(d-1)!}{d!}\phi(a_{[d+2]\backslash\{i\}}\mu)\\
&+\sum_{i=1}^{d+2}a_i\phi(a_i\mu)=0,
\end{align*}
which completes the proof of the first part of the proposition.

\medskip

\emph{2}. By Lemma~\ref{lemma:commutator lemma}, it is sufficient to check that for $d=2,3$ the following identity holds:
\begin{multline*}
\left\{\sum_{b_1+b_2+b_3=0}\frac{:p_{b_1}p_{b_2}p_{b_3}:}{6},-\sum_{a_1+\ldots+a_{d}=0}h_d(\oa,\mu)\frac{:p_{a_1}\cdots p_{a_{d}}:}{d!}\right\}+\\
+\left\{\sum_{a\ge 1}\phi(a\mu)p_{-a}p_{a},\sum_{a_1+\ldots+a_{d+1}=0}f_d(\oa,\mu)\frac{:p_{a_1}\cdots p_{a_{d+1}}:}{(d+1)!}\right\}=0,
\end{multline*}
or equivalently
$$
-\sum_{1\le i<j\le d+1}(a_i+a_j)h_d\lb a_1,\ldots,\widehat{a_i},\ldots,\widehat{a_j},\ldots,a_{d+1},a_i+a_j,\mu\rb+\lb\sum_{i=1}^{d+1}a_i\phi(a_i\mu)\rb f_d(\oa,\mu)=0.
$$

\medskip

\underline{The case $d=2$}. We have to check that
$$
\sum_{i=1}^3 a_i\theta(a_i\mu)+\frac{1}{2}\lb\sum_{i=1}^3 a_i\phi(a_i\mu)\rb\lb\sum_{i=1}^3\phi(a_i\mu)\rb=0.
$$
Since $\theta(x)=\frac{\phi(x)-\phi(x)^2}{2}-\frac{x^2}{24}$, this is equivalent to
\beq\label{eq:triple-phi identity}
\sum_{i=1}^3 a_i\phi(a_i)+\sum_{1\le i<j\le 3}(a_i+a_j)\phi(a_i)\phi(a_j)=\sum_{i=1}^3\frac{a_i^3}{12}.
\eeq
This relation is equivalent to the equality
$$
\sum_{1\le i<j\le 3}(a_i+a_j)B(a_i)B(a_j)=\sum_{i=1}^3\frac{a_i^3}{12},\quad B(x):=\phi(x)+1=\frac{x}{2}\frac{e^x+1}{e^x-1},
$$
which can be checked by a direct computation.

\medskip

\underline{The case $d=3$}. We have to check that
$$
\underbrace{-\sum_{1\le i<j\le 4}(a_i+a_j)h_3\lb a_1,\ldots,\widehat{a_i},\ldots,\widehat{a_j},\ldots,a_4,a_i+a_j,\mu\rb}_{S_1:=}+\underbrace{\lb\sum_{i=1}^4a_i\phi(a_i\mu)\rb f_3(\oa,\mu)}_{S_2:=}=0.
$$
Without loss of generality, we can assume that $\mu=1$. We have
\begin{align*}
&S_1=\sum_{i=1}^4a_i\lb\frac{\phi(a_i)}{2}-\frac{\phi(a_i)^2}{3}-\frac{a_i^2}{24}\rb - \frac{1}{6}\sum_{1\le i<j\le 4}(a_i+a_j)\big(\phi(a_i)\phi(a_j)+(\phi(a_i)+\phi(a_j))\phi(a_i+a_j)\big),\\
&S_2=\lb\sum_{i=1}^4a_i\phi(a_i)\rb\lb\sum_{i=1}^4\frac{\phi(a_i)}{3}+\sum_{1\le i<j\le 4}\frac{\phi(a_i+a_j)}{12}\rb,
\end{align*}
and therefore
\beq\label{eq:formula for S1+S2}
S_1+S_2=\sum_{i=1}^4\lb\frac{a_i\phi(a_i)}{2}-\frac{a_i^3}{24}\rb +\frac{1}{6}\sum_{1\le i<j\le 4}\big((a_i+a_j)\phi(a_i)\phi(a_j)-(a_i\phi(a_j)+a_j\phi(a_i))\phi(a_i+a_j)\big).
\eeq
For a pair of integers $1\le i<j\le 4$, let us write relation~\eqref{eq:triple-phi identity} for the variables $a_i,a_j,a_k+a_l$, where $\{i,j,k,l\}=\{1,2,3,4\}$, and then take the sum of these relations over all pairs $1\le i<j\le 4$. One can easily check that the resulting equality exactly says that the right-hand side of~\eqref{eq:formula for S1+S2} is zero.
\end{proof}

\medskip

As a corollary, we get
\begin{align}
&\int_{\oM_{g,3}}\DR_g(0,a,-a)\psi_1^2\lambda_g\lambda_{g-2}=(g+1)\frac{|B_{2g}|}{(2g)!}a^{2g},&& g\ge 2,\label{eq:psi2-lambdag-lambdag2}\\
&\int_{\oM_{g,4}}\DR_g(0,a,b,c)\psi_1^2\lambda_g\lambda_{g-1}=\frac{|B_{2g}|}{2(2g)!}(a^{2g}+b^{2g}+c^{2g}),&& g\ge 1.\notag
\end{align}

\medskip

\section{Proof of the conjecture for $d\le 3$}

\subsection{The operator $\oG^\bullet_2$}

By Proposition~\ref{proposition:leading terms of G_d}, we have
\beqn
\oG_2^\bullet=\hspace{-0.2cm}\sum_{a+b+c+d=0}\hspace{-0.2cm}\frac{:p_ap_bp_cp_d:}{24}+\frac{i\hbar-1}{4\mu}\hspace{-0.1cm}\sum_{a+b+c=0}\hspace{-0.1cm}\phi(a\mu):p_ap_bp_c:+\sum_{a\ge 1}\left(\frac{a^2}{12}i\hbar-\frac{(i\hbar-1)^2}{\mu^2}\theta(a\mu)\right)q_ap_a.
\eeqn
Since 
$$
\lim_{\mu\in\mbR,\,\mu\to+\infty}\frac{\phi(a\mu)}{\mu}=\frac{|a|}{2},\qquad \lim_{\mu\in\mbR,\,\mu\to+\infty}\frac{\theta(a\mu)}{\mu^2}=-\frac{a^2}{6},\qquad a\in\mbZ\backslash\{0\},
$$
we obtain
$$
\oG_2^{\lim}=\sum_{a+b+c+d=0}\frac{:p_ap_bp_cp_d:}{24}+\frac{i\hbar-1}{8}\sum_{a+b+c=0}|a|:p_ap_bp_c:+\sum_{a\ge 1}\left(\frac{a^2}{12}i\hbar+\frac{(i\hbar-1)^2}{6}a^2\right)q_a p_{a}.
$$

\medskip

Let us now prove the first equation in~\eqref{eq:I and G for d2 and d3}. Using~\eqref{eq:formula for I3}, and since
\begin{align*}
\sum_{i,j,k\geq 1}q_{i}p_{i-j}p_{j-k}p_k =&\sum_{i,j,k\geq 1}q_{i}:p_{i-j}p_{j-k}:p_k + \sum_{i>j\ge 1}q_{i}[p_{i-j},p_{j-i}]p_i=\\
=&\sum_{i,j,k\geq 1}:p_{-i}p_{i-j}p_{j-k}p_k: + \alpha\sum_{i\ge 1} \frac{i^2-i}{2}q_i p_i=\\
=&\frac{1}{4}\sum_{a+b+c+d=0}:p_a p_b p_c p_d: -\frac{1}{2}\sum_{a,b\ge 1}p_{-a} p_{-b}p_a p_b + \alpha\sum_{i\ge 1}\frac{i^2-i}{2}q_i p_i,
\end{align*}
we have
\begin{align*}
I^{(3)}=&\frac{1}{4}\sum_{a+b+c+d=0}:p_a p_b p_c p_d: + \frac{3(\alpha-1)}{2}\sum_{a,b\ge 1} (a+b)(p_{-a}p_{-b}p_{a+b}+p_{-a-b}p_ap_b)\\
&-\frac{1}{2}\sum_{a,b\ge 1}p_{-a} p_{-b}p_a p_b +\sum_{i\ge 1}\lb\alpha\frac{i^2-i}{2}+(\alpha-1)^2 i^2\rb p_{-i} p_i.
\end{align*}
Therefore,
$$
I^{(3)}-6\oG_2^{\lim}=-\frac{1}{2}\sum_{a,b\ge 1}p_{-a} p_{-b}p_a p_b -\frac{\alpha}{2}\sum_{i\ge 1}i p_{-i} p_i=-\frac{\oG_0^2}{2},
$$
which completes the proof of the first equation in~\eqref{eq:I and G for d2 and d3}.

\medskip

\subsection{The operator $\oG^\bullet_3$}\label{subsection:G3 Hamiltonian}

By Proposition~\ref{proposition:leading terms of G_d}, we have
\begin{align*}
    \oG'_3= &\sum_{\sum_{i=1}^5 a_i=0}\frac{:\prod p_{a_i}:}{5!} + (i\hbar-1)\sum_{\sum_{i=1}^4a_i=0}\left( \frac{\sum\phi(a_i\mu)}{3\mu} + \frac{\sum_{i<j}\phi((a_i+a_j)\mu)}{12\mu}\right)\frac{:\prod p_{a_i}:}{4!}\\
    &+ \sum_{\sum_{i=1}^{3}a_i=0}i\hbar\frac{\sum a_i^2 -1}{24}\frac{:\prod p_{a_i}:}{3!}\\
    &-\sum_{\sum_{i=1}^{3}a_i=0}(i\hbar-1)^2\left(\frac{\sum\phi(a_i\mu)}{4\mu^2}-\frac{\sum\phi(a_i\mu)^2}{6\mu^2}-\frac{\sum a_i^2}{48} - \frac{\sum_{i<j}\phi(a_i\mu)\phi(a_j\mu)}{6\mu^2}\right)\frac{:\prod p_{a_i}:}{3!}\\
    & + \oG_{3,2}',
\end{align*}
where
\begin{align*}
\oG_{3,2}'=&\lb 1-(i\hbar)^3\rb\sum_{g\geq 3}(-1)^g\mu^{2g-3}\sum_{a\ge 1}\left(\int_{\oM_{g,3}}\DR_g (0,a,-a)\lambda_g\lambda_{g-3}\psi_1^3\right) q_a p_a\\
&+i\hbar(1-i\hbar)\sum_{g\geq 2}(-1)^{g-1}\mu^{2g-3}\sum_{a\ge 1}\left(\int_{\oM_{g,3}}\DR_g (0,a,-a)\lambda_{g-1}\lambda_{g-2}\psi_1^3\right)q_a p_a.
\end{align*}

\medskip

Using \cite[Theorem 4]{BSSZ15}, we compute
\begin{align*}
&\int_{\oM_{g,3}}\DR_g (-a-b,a,b)\lambda_g\lambda_{g-3}\psi_1^3=\frac{1}{2g+1}\int_{\oM_{g,2}}\DR_g(-a-b,a+b)\lambda_g\lambda_{g-3} +\\
&\hspace{3cm}+\sum_{\substack{g_1\ge 1,\,g_2\ge 2\\ g_1+g_2=g}}\frac{2g_1}{2g+1}\frac{aP_{g_1}(a)+bP_{g_1}(b)}{a+b}\left( \int_{\oM_{g_2,3}}\DR_{g_2}(-a-b,a,b)\lambda_{g_2}\lambda_{g_2-2}\psi_1^2 \right),
\end{align*}
where we denote $P_g(a):=\int_{\oM_{g,2}}\DR_g(a,-a)\lambda_g\lambda_{g-1}=a^{2g}\frac{|B_{2g}|}{(2g)(2g)!}$. Using formula~\eqref{eq:psi2-lambdag-lambdag2} and the identity $\left.\frac{a^{2g_1+1}+b^{2g_1+1}}{a+b}\right|_{b=-a}=(2g_1+1)a^{2g_1}$, we obtain
\begin{equation*}
\int_{\oM_{g,3}}\DR_g (0,a,-a)\lambda_g\lambda_{g-3}\psi_1^3=a^{2g}\sum_{\substack{g_1\ge 1,\,g_2\ge 2\\ g_1+g_2=g}}\frac{(2g_1+1)(g_2+1)}{2g+1}\frac{|B_{2g_1}||B_{2g_2}|}{(2g_1)!(2g_2)!}.
\end{equation*}

\medskip

We proceed with another integral, again using~\cite[Theorem 4]{BSSZ15}:
\begin{align}
&\int_{\oM_{g,3}}\DR_g (-a-b,a,b)\lambda_{g-1}\lambda_{g-2}\psi_1^3=\notag\\
&\hspace{1.5cm}=\frac{1}{2g+1}\left( \int_{\oM_{g,2}}\DR_g(-a-b,a+b)\lambda_{g-1}\lambda_{g-2}\psi_1^2 \right)\label{eq:psi3-integral1}\\
&\hspace{1.9cm}+\sum_{\substack{g_1\ge 1,\,g_2\ge 2\\g_1+g_2=g}}\frac{6g_1}{2g+1}\frac{a P_{g_1}(a)+b P_{g_1}(b)}{a+b}
\left(\int_{\oM_{g_2,3} }\DR_{g_2}(-a-b,a,b)\lambda_{g_2}\lambda_{g_2-2}\psi_1^2 \right) \label{eq:psi3-integral1b}\\
&\hspace{1.9cm}+\frac{2(g-1)}{2g+1}\frac{a P_{g-1}(a)+b P_{g-1}(b)}{a+b}
\left(\int_{\oM_{1,3} }\DR_{1}(-a-b,a,b)\psi_1^2 \right)\label{eq:psi3-integral2}\\
&\hspace{1.9cm}+\sum_{\substack{i,j\ge 1\\i+j=a}}\frac{ij}{2(2g+1)(a+b)}
\left(\int_{\oM_{g-1,4} }\DR_{g-1}(-a-b,i,j,b)\lambda_{g-1}\lambda_{g-2}\psi_1^2 \right)\label{eq:psi3-integral3}\\
&\hspace{1.9cm}+\sum_{\substack{i,j\ge 1\\i+j=b}}\frac{ij}{2(2g+1)(a+b)}
\left(\int_{\oM_{g-1,4} }\DR_{g-1}(-a-b,a,i,j)\lambda_{g-1}\lambda_{g-2}\psi_1^2 \right),\label{eq:psi3-integral4}
\end{align}
Substituting $b=-a$ in~\eqref{eq:psi3-integral1}, we obtain
\begin{multline*}
\frac{1}{2g+1}\left( \int_{\oM_{g,2}}\DR_g(0,0)\lambda_{g-1}\lambda_{g-2}\psi_1^2 \right)=
\frac{(-1)^g}{2g+1}\left( \int_{\oM_{g,2}}\psi_1^2\lambda_g\lambda_{g-1}\lambda_{g-2} \right)=\\ 
=\frac{(-1)^g}{2g+1}\left( \int_{\oM_{g,1}}\psi_1\lambda_g\lambda_{g-1}\lambda_{g-2} \right)= \frac{(-1)^g}{2g+1}\frac{|B_{2g-2}||B_{2g}|}{4g(2g-2)!}.
\end{multline*}
Substituting $b=-a$ in~\eqref{eq:psi3-integral1b}, we obtain
$$
a^{2g}\sum_{\substack{g_1\ge 1,\,g_2\ge 2\\g_1+g_2=g}}\frac{3(2g_1+1)(g_2+1)}{2g+1}\frac{|B_{2g_1}||B_{2g_2}|}{(2g_1)!(2g_2)!}.
$$
Substituting $b=-a$ in~\eqref{eq:psi3-integral2}, we obtain
$$
\frac{2g-1}{2g+1}\frac{|B_{2g-2}|}{(2g-2)!}\frac{2a^2-1}{24}a^{2g-2}.
$$

\medskip

Let us now compute the terms~\eqref{eq:psi3-integral3} and~\eqref{eq:psi3-integral4}. Using that
\begin{align*}
&\int_{\oM_{g,4}}\DR_g(-a_1-a_2-a_3,a_1,a_2,a_3)\psi_1^2\lambda_g\lambda_{g-1}=\\
=&\frac{|B_{2g}|}{(2g+2)(2g)!}\lb\frac{\lb\sum a_i\rb^{2g}}{g(2g+1)}+\frac{\sum a_i^{2g+1}}{\sum a_i}+\sum_{\substack{i<j\\\{i,j,k\}=\{1,2,3\}}}\frac{(a_i+a_j)^{2g+2}+(a_i+a_j)a_k^{2g+1}}{(2g+1)\lb\sum a_i\rb^2}\rb,
\end{align*}
for the term~\eqref{eq:psi3-integral3} we obtain
\begin{align*}
\frac{|B_{2g-2}|}{2(2g+1)!}\sum_{i=1}^{a-1}&\left[\frac{i(a-i)(a+b)^{2g-3}}{g-1}+\frac{i(a-i)(2g-1)}{(a+b)^2}\lb \underline{i^{2g-1}+(a-i)^{2g-1}} + b^{2g-1}\rb\right.\\    
&+ \frac{i(a-i)a(a^{2g-1}+b^{2g-1})}{(a+b)^3} + \underline{\underline{\frac{i(a-i)(b+i)((b+i)^{2g-1}+(a-i)^{2g-1})}{(a+b)^3}}}\\
    & + \left.\underline{\underline{\frac{i(a-i)(a+b-i) ((a+b-i)^{2g-1}+i^{2g-1})}{(a+b)^3}}}\right],
\end{align*}
and for the term~\eqref{eq:psi3-integral4}:
\begin{align*}
\frac{|B_{2g-2}|}{2(2g+1)!}\sum_{i=1}^{b-1}&\left[\frac{i(b-i)(a+b)^{2g-3}}{g-1}+\frac{i(b-i)(2g-1)}{(a+b)^2}\lb \underline{i^{2g-1}+(b-i)^{2g-1}} + a^{2g-1}\rb\right.\\    
&+ \frac{i(b-i)b(b^{2g-1}+a^{2g-1})}{(a+b)^3} + \underline{\underline{\frac{i(b-i)(a+i)((a+i)^{2g-1}+(b-i)^{2g-1})}{(a+b)^3}}}\\
    & + \left.\underline{\underline{\frac{i(b-i)(a+b-i) ((a+b-i)^{2g-1}+i^{2g-1})}{(a+b)^3}}}\right].
\end{align*}

\medskip

The underlined terms give
$$
\frac{2C(2g-1)}{(a+b)^2}(aS_{2g}(a)+bS_{2g}(b)-S_{2g+1}(a)-S_{2g+1}(b)),
$$
where $C:=\frac{|B_{2g-2}|}{2(2g+1)!}$ and we use the notation 
$$
S_m(n):=\sum_{i=0}^{n-1} i^m=\sum_{k=0}^m\frac{m!}{k!(m+1-k)!}B_k n^{m+1-k},\quad n,m\ge 0.
$$ 
The doubly underlined terms give
\begin{align*}
&\frac{2C}{(a+b)^3}\Bigg[\sum_{i=b}^{a+b-1}(i-b)(a+b-i)i^{2g}+\sum_{i=0}^{a-1}(a-i)(a+b-i)i^{2g}\\
&+\sum_{i=a}^{a+b-1}(i-a)(a+b-i)i^{2g}+\sum_{i=0}^{b-1}(b-i)(a+b-i)i^{2g}\Bigg]=\\
=&\frac{2C}{(a+b)^3}\Bigg[\sum_{i=0}^{a+b-1}(i-b)(a+b-i)i^{2g}-2\sum_{i=0}^{b-1}(i-b)(a+b-i)i^{2g}\\
&+\sum_{i=0}^{a+b-1}(i-a)(a+b-i)i^{2g}-2\sum_{i=0}^{a-1}(i-a)(a+b-i)i^{2g}\Bigg]=\\
=&\frac{2C}{(a+b)^3}\Bigg[-2S_{2g+2}(a+b)+3(a+b)S_{2g+1}(a+b)-(a+b)^2S_{2g}(a+b)\\
&+2(S_{2g+2}(b)+S_{2g+2}(a))-2((a+2b)S_{2g+1}(b)+(b+2a)S_{2g+1}(a))+2(a+b)(bS_{2g}(b)+aS_{2g}(a))\Bigg].
\end{align*}
In total, the sum of the terms~\eqref{eq:psi3-integral3} and~\eqref{eq:psi3-integral4} is equal to
\begin{align*}
&\underbrace{C\frac{(a^3+b^3-a-b)(a+b)^{2g-3}}{6(g-1)}}_{S_1:=}\\
&+\underbrace{C\Bigg[\frac{(2g-1)(a+b)(a^{2g-1}(b^3-b) + b^{2g-1}(a^3-a))+(a^{2g-1}+b^{2g-1})(a^4+b^4-a^2-b^2)}{6(a+b)^3}\Bigg]}_{S_2:=}\\
&+\frac{2C}{(a+b)^3}\Bigg[-2S_{2g+2}(a+b)+3(a+b)S_{2g+1}(a+b)-(a+b)^2S_{2g}(a+b)\\
&+2\big(S_{2g+2}(a)+S_{2g+2}(b)\big)-\big(((2g+3)a+(2g+1)b)S_{2g+1}(a)+((2g+3)b+(2g+1)a)S_{2g+1}(b)\big)\\
&+(2g+1)(a+b)\big(aS_{2g}(a)+bS_{2g}(b)\big)\Bigg].
\end{align*}

\medskip

Let us now compute the sum of the terms~\eqref{eq:psi3-integral3} and~\eqref{eq:psi3-integral4} after the substitution $b=-a$. We have $S_1|_{b=-a}=0$ and
$$
S_2|_{b=-a}=C\Bigg[\lb-\frac{2}{9}g^3+\frac{4}{3}g^2-\frac{11}{18}g\rb a^{2g}+\lb\frac{2}{9}g^3-\frac{2}{3}g^2+\frac{5}{18}g\rb a^{2g-2}\Bigg].
$$
The rest of the terms give
\begin{align*}
&\frac{2C}{(a+b)^3}B_{2g}\Bigg[-2(a+b)^3\frac{(2g+2)(2g+1)}{6}+3(a+b)(a+b)^2\frac{2g+1}{2}-(a+b)^3\Bigg]\\
&+\frac{2C}{(a+b)^3}\sum_{k=0}^{2g}\frac{B_k}{k!}\Bigg[2\frac{(2g+2)!}{(2g+3-k)!}(a^{2g+3-k}+b^{2g+3-k})\\
&\hspace{2.3cm}-\frac{(2g+1)!}{(2g+2-k)!}\lb((2g+3)a+(2g+1)b)a^{2g+2-k}+((2g+3)b+(2g+1)a)b^{2g+2-k}\rb\\
&\hspace{2.3cm}+(2g+1)(a+b)\frac{(2g)!}{(2g+1-k)!}(a^{2g+2-k}+b^{2g+2-k})\Bigg]\Bigg|_{b=-a}=\\
=&\frac{|B_{2g-2}|B_{2g}}{(2g+1)!}\frac{(1-2g)(4g-1)}{6}+\frac{|B_{2g-2}|}{6}\sum_{k=0}^{2g}\frac{B_k}{k!}\frac{1-k}{(2g-k)!}a^{2g-k}=\\
=&-\frac{|B_{2g-2}|B_{2g}}{2(2g+1)(2g-2)!}+\frac{|B_{2g-2}|}{6}\sum_{k=0}^{2g-2}\frac{B_k}{k!}\frac{1-k}{(2g-k)!}a^{2g-k}.
\end{align*}
So the sum of the terms~\eqref{eq:psi3-integral3} and~\eqref{eq:psi3-integral4} after the substitution $b=-a$ is equal to
\begin{align*}
&\frac{|B_{2g-2}|}{2(2g+1)!}\Bigg[\lb-\frac{2}{9}g^3+\frac{4}{3}g^2-\frac{11}{18}g\rb a^{2g}+\lb\frac{2}{9}g^3-\frac{2}{3}g^2+\frac{5}{18}g\rb a^{2g-2}\Bigg]\\
&-\frac{|B_{2g-2}|B_{2g}}{2(2g+1)(2g-2)!}+\frac{|B_{2g-2}|}{6}\sum_{k=0}^{2g-2}\frac{B_k}{k!}\frac{1-k}{(2g-k)!}a^{2g-k}.
\end{align*}

\medskip

We finally get
\begin{align*}
\int_{\oM_{g,3}}\DR_g (0,-a,a)&\lambda_{g-1}\lambda_{g-2}\psi_1^3=\\
=&\frac{(-1)^g|B_{2g-2}||B_{2g}|}{4g(2g-2)!} + a^{2g}\sum_{\substack{g_1\ge 1,\,g_2\ge 2\\ g_1+g_2=g}}\frac{3(2g_1+1)(g_2+1)}{(2g+1)}\frac{|B_{2g_1}||B_{2g_2}|}{(2g_1)!(2g_2)!}\\
&+\frac{a^{2g-2}(5a^2-2)}{72}\frac{|B_{2g-2}|}{(2g-2)!}+\frac{|B_{2g-2}|}{6}\sum_{k=0}^{2g-2}\frac{B_k}{k!}\frac{1-k}{(2g-k)!}a^{2g-k}
\end{align*}
and
\begin{align}
\oG_{3,2}'=&(1-i\hbar)^3\sum_{a>0}\sum_{g\geq 3}a^{2g}\mu^{2g-3}\left(\sum_{\substack{g_1\ge 1,\,g_2\ge 2\\g_1+g_2=g}}\frac{(2g_1+1)(g_2+1)}{2g+1}\frac{B_{2g_1}B_{2g_2}}{(2g_1)!(2g_2)!}\right) q_{a} p_a\notag\\
& +i\hbar(1-i\hbar)\sum_{a>0}\sum_{g\geq 2}\mu^{2g-3}\frac{B_{2g-2}B_{2g}}{4g(2g-2)!} q_{a} p_a\label{eq:divergent term}\\
& +i\hbar(1-i\hbar)\sum_{a>0}\sum_{g\geq 2}\mu^{2g-3}\frac{a^{2g-2}(2-5a^2)}{72}\frac{B_{2g-2}}{(2g-2)!} q_{a} p_a\notag\\
& +i\hbar(1-i\hbar)\sum_{a>0}\sum_{g\geq 2}\mu^{2g-3}\frac{B_{2g-2}}{6}\sum_{k=0}^{2g-2}\frac{B_k}{k!}\frac{k-1}{(2g-k)!}a^{2g-k}q_{a} p_a. \label{eq:the last term of quadratic terms}
\end{align}

\medskip

We compute
\begin{align*}
&\sum_{g\ge 3}\frac{x^{2g}}{2g+1}\sum_{\substack{g_1\ge 1,\,g_2\ge 2\\g_1+g_2=g}}(2g_1+1)(g_2+1)\frac{B_{2g_1}}{(2g_1)!}\frac{B_{2g_2}}{(2g_2)!}=\frac{1}{x}\int\frac{d}{dx}\lb x\phi(x)\rb\cdot\theta(x)dx=\\
&\hspace{2cm}=-\frac{p(x) + 3 e^x q(x) -3 e^{2x}q(-x)-e^{3x}p(-x)}{72 (e^x-1)^3},\quad \begin{minipage}{5cm}
$p(x)=72 + 63 x + 22 x^2 + 3 x^3$,\\
$q(x)=-72 - 21 x + x^3$.
\end{minipage}
\end{align*}
Let us analyze the sum in~\eqref{eq:the last term of quadratic terms}. Note that for fixed $a\ge 1$ we have
$$
\sum_{k=0}^{2g-2}\frac{B_k}{k!}\frac{(2g-2)!}{(2g-k)!}(k-1)a^{2g-k}=S_{2g-1}(a)-aS_{2g-2}(a)=\sum_{i=1}^{a-1}(i-a)i^{2g-2},\quad g\ge 2,
$$
and therefore
\begin{align*}
\sum_{g\geq 2}\mu^{2g-3}\frac{B_{2g-2}}{6}\sum_{k=0}^{2g-2}\frac{B_k}{k!}\frac{k-1}{(2g-k)!}a^{2g-k}=&\frac{1}{6}\sum_{i=1}^{a-1}(i-a)\sum_{g\ge 2}\frac{B_{2g-2}}{(2g-2)!}i^{2g-2}\mu^{2g-3}=\frac{1}{6\mu}\sum_{i=1}^{a-1}(i-a)\phi(\mu i).
\end{align*}

\medskip

\begin{remark}\label{remark:divergent term}
We see that all the terms in the expression for $\oG_3'$ have a limit when $\mu\in\mbR$ and $\mu\to+\infty$, except the term~\eqref{eq:divergent term}, which is equal to $-i\hbar(1-i\hbar)\sum_{a>0}f_{1,2}(\mu)q_{a} p_a$. 
\end{remark}

\medskip

So for $\oG_3^\bullet$ we finally get the expression
\begin{align*}
    \oG^\bullet_3= &\sum_{\sum_{i=1}^5 a_i=0}\frac{:\prod p_{a_i}:}{5!} + (i\hbar-1)\sum_{\sum_{i=1}^4a_i=0}\left( \frac{\sum\phi(a_i\mu)}{3\mu} + \frac{\sum_{i<j}\phi((a_i+a_j)\mu)}{12\mu}\right)\frac{:\prod p_{a_i}:}{4!}\\
    &+ \sum_{\sum_{i=1}^{3}a_i=0}i\hbar\frac{\sum a_i^2}{24}\frac{:\prod p_{a_i}:}{3!}\\
    &-\sum_{\sum_{i=1}^{3}a_i=0}(i\hbar-1)^2\left(\frac{\sum\phi(a_i\mu)}{4\mu^2}-\frac{\sum\phi(a_i\mu)^2}{6\mu^2}-\frac{\sum a_i^2}{48} - \frac{\sum_{i<j}\phi(a_i\mu)\phi(a_j\mu)}{6\mu^2}\right)\frac{:\prod p_{a_i}:}{3!}
    \end{align*}
\begin{align*}
    & + (i\hbar-1)^3\sum_{a>0}\frac{p(a\mu) + 3 e^{a\mu} q(a\mu) -3 e^{2a\mu}q(-a\mu)-e^{3a\mu}p(-a\mu)}{72 (e^{a\mu}-1)^3\mu^3} q_a p_a\\
& +i\hbar(i\hbar-1)\sum_{a>0}\lb\frac{1+5a^2}{72}\frac{\phi(a\mu)}{\mu}+\sum_{i=1}^{a-1}\frac{a-i}{6}\frac{\phi(\mu i)}{\mu}\rb q_{a} p_a.
\end{align*}
The limit is equal to
\begin{align*}
    \oG^{\lim}_3= &\sum_{\sum_{i=1}^5 a_i=0}\frac{:\prod p_{a_i}:}{5!} + (i\hbar-1)\sum_{\sum_{i=1}^4a_i=0}\left(\frac{\sum|a_i|}{6} + \frac{\sum_{i<j}|a_i+a_j|}{24}\right)\frac{:\prod p_{a_i}:}{4!}\\
    &+ \sum_{\sum_{i=1}^{3}a_i=0}\lb i\hbar\frac{\sum a_i^2}{24}+(i\hbar-1)^2\left(\frac{\sum a_i^2}{16} + \frac{\sum_{i<j}|a_ia_j|}{24}\right)\rb\frac{:\prod p_{a_i}:}{3!}\\
    & + \sum_{a>0}\lb (i\hbar-1)^3\frac{a^3}{24}+i\hbar(i\hbar-1)\frac{7a^3-a}{144}\rb q_a p_a.
\end{align*}

\medskip

Let us now check the second equation in~\eqref{eq:I and G for d2 and d3}. We will check that
\beq\label{eq:pip equation}
\pi_k\lb I^{(4)}\rb=\pi_k\lb 4!\oG_3^{\lim}- 2 \oG_0^{\lim} \oG_1^{\lim}\rb
\eeq
separately for $2\le k\le 5$.

\medskip

\underline{Proof of~\eqref{eq:pip equation} for $k=5$}. Since both sides of~\eqref{eq:pip equation}, for $k=5$, have the form $\sum_{\substack{\lambda,\mu\in\mcP,\,|\lambda|=|\mu|\\l(\lambda)+l(\mu)=5}}c_{\lambda,\mu}q_\lambda p_\mu$, $c_{\lambda,\mu}\in\mbQ$, and we have already proved equation~\eqref{eq:pip equation} when $i\hbar=1$, equation~\eqref{eq:pip equation} is true for $k=5$.

\medskip

\underline{Proof of~\eqref{eq:pip equation} for $k=4$}. We compute
\begin{align}
\sum_{i,j,k\geq 1}(i+j+k)q_{i}p_{i-j}p_{j-k}p_k =&\sum_{i,j,k\geq 1}(i+j+k)q_{i}:p_{i-j}p_{j-k}:p_k + \sum_{i>j\ge 1}(2i+j)q_{i}[p_{i-j},p_{j-i}]p_i=\label{eq:normal ordering for 4 terms}\\
=&\sum_{i,j,k\geq 1}:p_{-i}p_{i-j}p_{j-k}p_k: + \alpha\sum_{i\ge 1} \frac{i(i-1)(7i+1)}{6}q_i p_i.\notag
\end{align}
Therefore, 
$$
\pi_4\lb I^{(4)}\rb=(\alpha-1)\sum_{i,j,k\geq 1}(i+j+k):q_i p_{i-j}p_{j-k}p_k:,
$$
while
\begin{align*}
&\pi_4\lb 4!\oG_3^{\lim}\rb=(\alpha-1)\sum_{\sum_{i=1}^4a_i=0}
c(a_1,a_2,a_3,a_4):\prod p_{a_i}:,\\
&\pi_4\lb- 2 \oG_0^{\lim} \oG_1^{\lim}\rb=-(\alpha-1)\sum_{a,b>0}aq_aq_bp_ap_b,
\end{align*}
where $c(a_1,a_2,a_3,a_4):=\frac{\sum|a_i|}{6}+\frac{\sum_{i<j}|a_i+a_j|}{24}$. So we have to check that
\beq\label{eq:identity for I and G - k4}
\sum_{i,j,k\geq 1}(i+j+k):q_i p_{i-j}p_{j-k}p_k:=\sum_{\sum_{i=1}^4a_i=0}
c(a_1,a_2,a_3,a_4):\prod p_{a_i}:-\sum_{a,b>0}aq_aq_bp_ap_b.
\eeq

\medskip

We split the left-hand side of~\eqref{eq:identity for I and G - k4} as follows:
\begin{align*}
&\underbrace{\sum_{i>j>k\geq 1}(i+j+k):q_i p_{i-j}p_{j-k}p_k:}_{S_1:=}+\underbrace{\sum_{1\le i<j<k}(i+j+k):q_i p_{i-j}p_{j-k}p_k:}_{S_2:=}\\
&+\underbrace{\sum_{\substack{i,j,k\geq 1\\i>j<k}}(i+j+k):q_i p_{i-j}p_{j-k}p_k:+\sum_{\substack{i,j,k\geq 1\\i<j>k}}(i+j+k):q_i p_{i-j}p_{j-k}p_k:}_{S_3:=}.
\end{align*}
We compute
\begin{align*}
&S_1=\sum_{a,b,c\geq 1}(a+2b+3c)q_{a+b+c} p_{a}p_{b}p_c=2\sum_{a,b,c\geq 1}(a+b+c)q_{a+b+c} p_{a}p_{b}p_c,\\
&S_2=\sum_{a,b,c\geq 1}(3a+2b+c)q_{a} q_b q_c p_{a+b+c}=2\sum_{a,b,c\geq 1}(a+b+c)q_{a} q_b q_c p_{a+b+c},
\end{align*}
and then easily check that
$$
S_1+S_2=\sum_{\substack{\sum_{i=1}^4a_i=0\\|\{i|a_i>0\}|=1\text{ or }3}}
c(a_1,a_2,a_3,a_4):\prod p_{a_i}:.
$$
So it remains to prove that
\beq\label{eq:identity for S3 for k4}
S_3=\sum_{\substack{\sum_{i=1}^4a_i=0\\|\{i|a_i>0\}|=2}}
c(a_1,a_2,a_3,a_4):\prod p_{a_i}:-\sum_{a,b>0}aq_aq_bp_ap_b.
\eeq

\medskip

We compute
\begin{align}
&S_3=\sum_{\substack{a,b,c,d\ge 1,\,a>c\\a+b=c+d}}(2a-c+d)q_aq_bp_cp_d+\sum_{\substack{a,b,c,d\ge 1\\a+b=c+d}}(2a+b+d)q_a q_{b}p_{c}p_d=\notag\\
&=\hspace{-0.2cm}\sum_{\substack{a,b,c,d\ge 1\\a+b=c+d\\d<a>c}}\hspace{-0.2cm}(6a+4b)q_aq_bp_cp_d+\hspace{-0.2cm}\sum_{\substack{a,b,c,d\ge 1\\a+b=c+d\\d>a>c}}\hspace{-0.2cm}(4c+6d)q_aq_bp_cp_d+\sum_{a>b\ge 1}(11a+7b)q_aq_bp_bp_a+\sum_{a\ge 1}4a q_aq_ap_ap_a.\label{eq:4 terms for S3}
\end{align}
Similarly, we decompose the right-hand side of~\eqref{eq:identity for S3 for k4} into the following sum:
\begin{align*}
&\sum_{\substack{\sum_{i=1}^4a_i=0\\a_1,a_2<0,\,a_3,a_4>0\\\max(|a_1|,|a_2|)>\max(|a_3|,|a_4|)}}\hspace{-0.5cm}6\,c(a_1,a_2,a_3,a_4):\prod p_{a_i}:+\hspace{-0.5cm}\sum_{\substack{\sum_{i=1}^4a_i=0\\a_1,a_2<0,\,a_3,a_4>0\\\max(|a_1|,|a_2|)<\max(|a_3|,|a_4|)}}\hspace{-0.5cm}6\,c(a_1,a_2,a_3,a_4):\prod p_{a_i}:+\\
&+\left[\sum_{\substack{a_1,a_2<0,\,a_3,a_4>0\\a_1\ne a_2,\{-a_1,-a_2\}=\{a_3,a_4\}}}6\,c(a_1,a_2,a_3,a_4):\prod p_{a_i}:-\sum_{a,b\ge 1,\,a\ne b}aq_aq_bp_ap_b\right]+\\
&+\left[\sum_{-a_1=-a_2=a_3=a_4\ge 1}6\,c(a_1,a_2,a_3,a_4):\prod p_{a_i}:-\sum_{a\ge 1}aq_aq_ap_ap_a\right].
\end{align*}
It is easy to check that these four summands are equal, respectively, to the four summands in line~\eqref{eq:4 terms for S3}, and, thus, equation~\eqref{eq:identity for S3 for k4} is proved.

\medskip

\underline{Proof of~\eqref{eq:pip equation} for $k=3$}. We have
$$
\pi_3\lb I^{(4)}\rb=\pi_3\lb\sum_{i,j,k,l\ge 1}q_i p_{i-j}p_{j-k}p_{k-l}p_l\rb+(\alpha-1)^2\sum_{i,j\geq 1}(i^2+ij+j^2)q_{i}p_{i-j}p_j.
$$
and 
\begin{align*}
&\pi_3\lb 4!\oG_3^{\lim}\rb=\sum_{a,b\geq 1}\left( \alpha (a^2 + ab + b^2) + (\alpha-1)^2(2(a^2+b^2)+3ab) \right)(q_{a+b}p_a p_b + q_a q_b  p_{a+b}),\\
&\pi_3\lb-2\oG_0^{\lim}\oG_1^{\lim}\rb=-\sum_{a,b\geq 1} \alpha (a + b) (q_{a+b}p_a p_b + q_a q_b  p_{a+b}).
\end{align*}
Since equation $\pi_3\lb I^{(4)}\rb=\pi_3\lb 4!\oG_3^{\lim}- 2 \oG_0^{\lim} \oG_1^{\lim}\rb$ is proved for $\alpha=1$, we conclude that
$$
\pi_3\lb\sum_{i,j,k,l\ge 1}q_i p_{i-j}p_{j-k}p_{k-l}p_l\rb=\alpha\sum_{a,b\geq 1}\left(a^2 + ab + b^2-a-b\right)(q_{a+b}p_a p_b + q_a q_b  p_{a+b}).
$$
It remains to check that
\beq\label{eq:identity for I and G - k3}
\sum_{i,j\geq 1}(i^2+ij+j^2)q_{i}p_{i-j}p_j=\sum_{a,b\geq 1} \left(2(a^2+b^2)+3ab \right)(q_{a+b}p_a p_b + q_a q_b  p_{a+b}).
\eeq
Indeed, we transform the left-hand side as follows:
\begin{align*}
&\sum_{i>j\geq 1}(i^2+ij+j^2)q_{i}p_{i-j}p_j+\sum_{1\le i<j}(i^2+ij+j^2)q_{i}p_{i-j}p_j=\\
=&\sum_{k,j\geq 1}((j+k)^2+(j+k)j+j^2)q_{j+k}p_{k}p_j+\sum_{i,k\geq 1}(i^2+i(i+k)+(i+k)^2)q_{i}q_k p_{i+k}=\\
=&\sum_{a,b\geq 1} \left(2(a^2+b^2)+3ab \right)(q_{a+b}p_a p_b + q_a q_b  p_{a+b}),
\end{align*}
and, thus, equation~\eqref{eq:identity for I and G - k3} is proved.

\medskip

\underline{Proof of~\eqref{eq:pip equation} for $k=2$}. By~\eqref{eq:normal ordering for 4 terms}, on the left-hand side we have
$$
\pi_2(I^{(4)})=\sum_{i\geq 1}\left((\alpha-1)^3 i^3 +
    \alpha(\alpha-1)\frac{7i^3-6i^2-i}{6}\right)q_{i}p_i,
$$
while on the right-hand side we have
\begin{align*}
&\pi_2\lb 4!\oG_3^{\lim}\rb= \sum_{i>0}\left((\alpha-1)^3i^3 + \alpha(\alpha-1)\frac{7i^3-i}{6}  \right)q_{i}p_i,\\
&\pi_2\lb-2\oG_0^{\lim}\oG_1^{\lim}\rb=-\sum_{i>0}\alpha(\alpha -1)i^2 q_i p_i.
\end{align*}
This completes the proof of equation~\eqref{eq:pip equation}.


\medskip

\begin{thebibliography}{BDGR20}

\bibitem[Blo22]{Blo22} X. Blot. {\it The quantum Witten--Kontsevich series and one-part double Hurwitz numbers}. Geom. Topol.~26 (2022), no.~4, 1669--1743.

\smallskip

\bibitem[BB24]{BB24} X. Blot, A. Buryak. {\it Quantum intersection numbers and the Gromov--Witten invariants of $\mathbb{CP}^1$}. Lett. Math. Phys. 114 (2024), no. 6, paper no.~131.

\smallskip

\bibitem[BLS24]{BLS24} X. Blot, D. Lewa\'nski, S. Shadrin. {\it Cohomological representations of quantum tau functions}. arXiv:2411.03499.

\smallskip

\bibitem[Bur15a]{Bur15a} A. Buryak. {\it Double ramification cycles and integrable hierarchies}. Comm. Math. Phys.~336 (2015), no.~3, 1085--1107.

\smallskip

\bibitem[Bur15b]{Bur15b} A. Buryak. {\it Dubrovin-Zhang hierarchy for the Hodge integrals}. Commun. Number Theory Phys.~9 (2015), no.~2, 239--271. 

\smallskip

\bibitem[BR16]{BR16} A. Buryak, P. Rossi. {\it Double ramification cycles and quantum integrable systems}. Lett. Math. Phys. 106 (2016), no. 3, 289--317.

\smallskip

\bibitem[BR25]{BR25} A. Buryak, P. Rossi. {\it Deformations of the Riemann hierarchy and the geometry of $\oM_{g,n}$}. Int. Math. Res. Not. IMRN (2025), no. 20, rnaf315.

\smallskip

\bibitem[BSSZ15]{BSSZ15} A. Buryak, S. Shadrin, L. Spitz, D. Zvonkine. {\it Integrals of $\psi$-classes over double ramification cycles}. Amer. J. Math. 137 (2015), no. 3, 699--737.

\smallskip

\bibitem[BXY26]{BXY26} A. Buryak, J. Xu, D. Yang. {\it Bihamiltonian tests for integrable systems associated to rank-$1$ F-CohFTs}. arXiv:2601.13203.

\smallskip

\bibitem[dCR19]{dCR19} A. du Crest de Villeneuve, P. Rossi. {\it Quantum $D_4$ Drinfeld--Sokolov hierarchy and quantum singularity theory}. J. Geom. Phys.~141 (2019), 29--44.

\smallskip

\bibitem[BDGR20]{BDGR20} A. Buryak, B. Dubrovin, J. Gu\'er\'e, P. Rossi. {\it Integrable systems of double ramification type}. Int. Math. Res. Not. IMRN (2020), no.~24, 10381--10446.

\smallskip

\bibitem[Eli07]{Eli07} Y. Eliashberg. {\it Symplectic field theory and its applications}. International Congress of Mathematicians. Vol. I, 217--246. European Mathematical Society (EMS), Z\"urich, 2007.

\smallskip

\bibitem[FP00]{FP00} C. Faber, R. Pandharipande. {\it Logarithmic series and Hodge integrals in the tautological ring}. With an appendix by Don Zagier. Dedicated to William Fulton on the occasion of his 60th birthday. Michigan Math. J.~48 (2000), 215--252.

\smallskip

\bibitem[FP03]{FP03} C. Faber, R. Pandharipande. {\it Hodge integrals, partition matrices, and the $\lambda_g$ conjecture}. Ann. of Math. (2) 157 (2003), no.~1, 97--124.

\smallskip

\bibitem[vIR24]{vIR24} J.-W. M. van Ittersum, G. Ruzza. {\it Quantum KdV hierarchy and quasimodular forms}. Commun. Number Theory Phys.~18 (2024), no.~2, 405--439.

\smallskip

\bibitem[vIR25]{vIR25} J.-W. M. van Ittersum, G. Ruzza. {\it Quantum KdV hierarchy and shifted symmetric functions}. Int. Math. Res. Not. IMRN 2025, no.~9, paper no. rnaf102. 

\smallskip

\bibitem[Kon92]{Kon92} M. Kontsevich. {\it Intersection theory on the moduli space of curves and the matrix Airy function}. Comm. Math. Phys.~147 (1992), no.~1, 1--23.

\smallskip

\bibitem[NS13]{NS13} M. Nazarov, E. Sklyanin. {\it Integrable hierarchy of the quantum Benjamin--Ono equation}. SIGMA Symmetry Integrability Geom. Methods Appl. 9 (2013),~078.

\smallskip

\bibitem[Pog03]{Pog03} A. K. Pogrebkov. {\it The boson--fermion correspondence and quantum integrable and dispersionless models}. Russian Math. Surveys 58 (2003), no.~5, 1003--1037.

\smallskip

\bibitem[RY21]{RY21} G. Ruzza, D. Yang. {\it On the spectral problem of the quantum KdV hierarchy}. J. Phys. A~54 (2021), no.~37, paper no.~374001.

\smallskip

\bibitem[Spe26]{Spe26} P. Spelier. {\it Polynomiality of the double ramification cycle}. Algebraic Geometry 13 (2026), no.~4, 588--597.

\smallskip

\bibitem[Wit91]{Wit91} E. Witten. {\it Two-dimensional gravity and intersection theory on moduli space}. Surveys in differential geometry (Cambridge, MA, 1990), 243--310, Lehigh Univ., Bethlehem, PA, 1991.

\end{thebibliography}
\end{document}